\documentclass[submission,copyright,creativecommons]{eptcs}
\providecommand{\event}{SOS 2007} % Name of the event you are submitting to

\usepackage{iftex}

\ifpdf
  \usepackage{underscore}         % Only needed if you use pdflatex.
  \usepackage[T1]{fontenc}        % Recommended with pdflatex
\else
  \usepackage{breakurl}           % Not needed if you use pdflatex only.
\fi
\usepackage{graphicx}
\usepackage{amsmath}
\usepackage{amsthm}

\newtheorem{theorem}{Theorem}[section]

\title{Theorem Discovery Amongst Cyclic Polygons}\author{Philip Todd
\institute{Saltire Software\\ Portland OR USA}
\email{philt@saltire.com}
}
\def\titlerunning{Theorem Discovery}
\def\authorrunning{P.H. Todd}
\begin{document}
\maketitle

\begin{abstract}
We examine a class of geometric theorems on cyclic $2n$-gons.  We prove that if we take n disjoint pairs of sides, each pair separated by an even number of polygon sides, then there is a linear combination of the angles between those sides which is constant.  We present a formula for the linear combination, which provides a theorem statement in terms of those angles.  We describe a program which uses this result to generate new geometry proof problems and their solutions.
\end{abstract}
\section{Introduction}
In \cite{AngleTheoremDiscovery}, a characterization is made of linear systems involving angle bisection conditions which are not full rank.  In such a system, one of the conditions is implied by the remainder, and, if the angle bisections are interpreted geometrically, this dependence may be stated in a number of different ways as a geometry theorem.  The characterization leads to a catalog of such linear systems.  An approach to theorem discovery is proposed wherein a linear system is initially selected, and then interpreted geometrically as a theorem.
In \cite{GeometryProofProblems}, a program is described which applies this approach, constructing a particular geometry theorem corresponding to a randomly selected  linear system from the catalog.  In order to reduce diagram complexity, the program is biased in favor of constructing cyclic polygons wherever possible.  In the case where it is able to construct a cyclic polygon using all the rows of the linear system, the theorem which is produced has the following form. Given a cyclic $2n$-gon, where $n-1$ specified pairs of sides are parallel, then a final specified pair of sides is also parallel.  For example, in a cyclic hexagon, with two pairs of opposite sides parallel, the third pair of sides is also parallel. (In passing, we note thet this theorem is not true for a cyclic octagon, but is for a cyclic decagon.)

In this presentation, we consider a generalization of the above class of theorems, where instead of making line pairs parallel, we allow line pairs to be given non zero, but determined angles.  While there is no geometric meaning to a parallel relationship between consecutive sides, replacing the parallelism by a defined non-zero angle permits adjacent sides of the polygon to be related. We examine the case of a $2n$-gon, with the angles between $n$ pairs of sides named.  We show that if the pairs share no side, and if the sides in each pair are either adjacent or seperated by an even number of polygon sides, then the named angles satisfy a particular linear relation.    

An approach to theorem discovery in this context mirrors and illustrates that described in  \cite{AngleTheoremDiscovery} and \cite{GeometryProofProblems}.  As a first step, a value of $n$ is chosen, and a set of line pairs conforming to our criterion selected from a catalog containing all such sets.  A diagram is produced directly which allows the coefficients of the constant linear combination to be computed.  

\section{Cyclic Polygons and Angle Bisectors}
\begin{figure}[h]%
\centering
\includegraphics[width=0.5\textwidth]{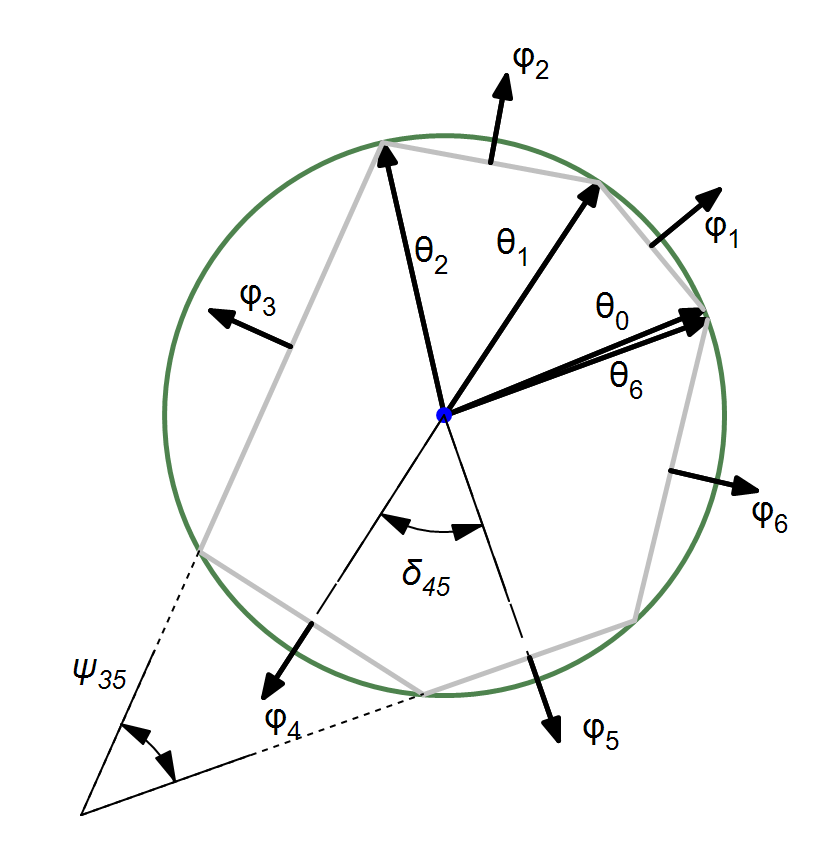}
\caption{$\theta_0 \ldots \theta_6$ are angles of position vectors. $\theta_6=\theta_1+2\pi$. $\phi_i=\frac{1}{2}(\theta_{i-1}+\theta_i)$  }
\label{notation}
\end{figure}

Let $P$ be a polygon whose vertices lie on the unit circle centered at the origin whose vertices $p_1,\ldots ,p_n$ have position vectors  $u_1, \ldots ,u_n$.  

We define $u_0=u_n$.

Let $\alpha(u,v)$ be the directed angle between vector $u$ and vector $v$. 
We define $\theta_0=0$ and for $i$ from $1$ to $n$:
$$\theta_i=\theta_{i-1}+\alpha(u_{i-1} , u_i)$$
As $u_0=u_n$, $\theta_n=\theta_0+2\pi W$ where $W$ is the winding number of the polygon about the origin.
For $i$ from $1$ to $n$ we define 
$$\phi_i=\frac {1}{2}(\theta_i+\theta_{i-1})$$
For $i$ and $j$ from $1$ to $n$ we define
$$\delta_{ij}=\phi_j-\phi_i$$
Define  $L_i$ to be the line passing through points 
$p_{i-1}$ and $p_i$ .  We will define $q_{ij}$ to be the intersection of  $L_i$ and $L_j$.  We define the angle 
$$\psi_{ij}=\angle p_{i-1}q_{ij}p_j$$  
\subsection{Cyclic Quadrilateral}
\begin{figure}[h]%
\centering
\includegraphics[width=0.55\textwidth]{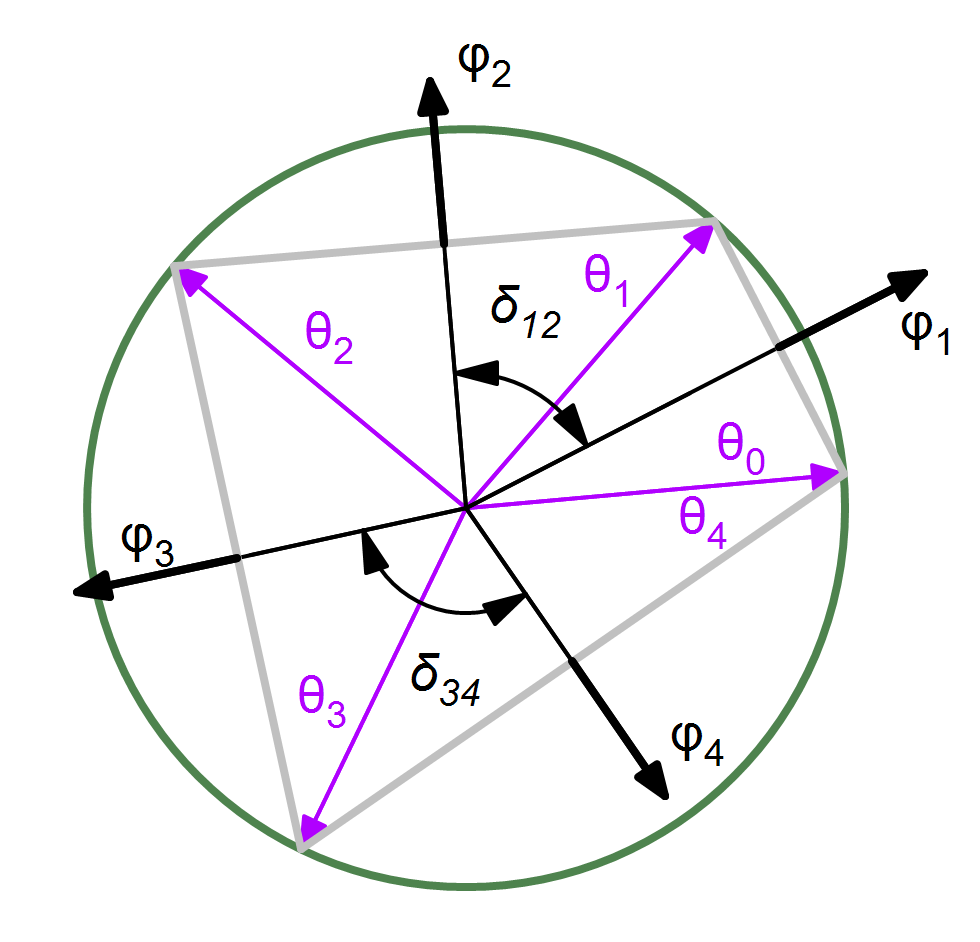}
\caption{A cyclic quadrilateral with winding number 1. }
\label{CyclicQuad1}
\end{figure}

We first examine the cyclic quadrilateral (Figure \ref{CyclicQuad1}). 

In the figure, the quadrilateral has winding number 1 about the circle center, hence $\theta_4=\theta_0+2\pi$.  The two indicated opposite angles of the quadrilateral have values 
$\pi-\delta_{12}$ and $\pi-\delta_{34}.$

The figure may be expressed by the following matrix equation:
\begin{equation}
\begin{pmatrix}
1 & 1 & 0 & 0 & 0 & -2 & 0 & 0 & 0  \\
0 & 1 & 1 & 0 & 0 & 0 & -2 & 0 & 0  \\
0 & 0 & 1 & 1 & 0 & 0 & 0 & -2 & 0   \\
0 & 0 & 0 & 1 & 1 & 0 & 0 & 0 & -2    \\
-1 & 0 & 0 & 0 & 1 & 0 & 0 & 0 & 0    \\
0 & 0 & 0 & 0 & 0 & -1 & 1 & 0 & 0    \\
0 & 0 & 0 & 0 & 0 & 0  & 0 & -1 & 1    \\
\end{pmatrix}
\begin{pmatrix}
\theta_0  \\
\theta_1  \\
\theta_2   \\
\theta_3  \\
\theta_4  \\
\phi_1   \\
\phi_2   \\
\phi_3   \\
\phi_4   \\
\end{pmatrix}
=
\begin{pmatrix}
0  \\
0  \\
0  \\
0  \\
2\pi   \\
\delta_{12}  \\
\delta_{34}  \\
\end{pmatrix}
\end{equation}
We transform the matrix equation by performing the following row operations:

$R1\leftarrow R1+R5$, $R2\leftarrow R2+2\times R6$, and $R3\leftarrow R3+2\times R7$, which, after eliminating the zero columns, gives this matrix equation:

\begin{equation} \label{eq:matrixeq}
\begin{pmatrix}
1 & 0 & 0 & 1 & -2 &  0   \\
1 & 1 & 0 & 0 &-2 & 0   \\
0 & 1 & 1 & 0 & 0 & -2    \\
0 & 0 & 1 & 1 & 0 & -2     \\
\end{pmatrix}
\begin{pmatrix}
\theta_1  \\
\theta_2   \\
\theta_3  \\
\theta_4  \\
\phi_1   \\
\phi_3   \\
\end{pmatrix}
=
\begin{pmatrix}
2\pi  \\
2\delta_{12} \\
0  \\
2\delta_{34}  \\
\end{pmatrix}
\end{equation}
The matrix can be triangularized using the algorithm of \cite{Tupero}.  Let $R_i$ be the i'th row of the original matrix, and $T_i$ the i'th row of the triangularized matrix, then 
$T_1=R_1$, and for $i$ from $2$ to $n$:
\begin{equation}
T_i=R_i-T_{i-1}
\end{equation}
Using this algorithm, our triangularized matrix equation is:
\begin{equation} \label{eq:triangularmatrix}
\begin{pmatrix}
1 & 0 & 0 & 1 & -2 &  0   \\
0 & 1 & 0 & -1 & 0 & 0   \\
0 & 0 & 1 & 1 & 0 & -2    \\
0 & 0 & 0 & 0 & 0 & 0     \\
\end{pmatrix}
\begin{pmatrix}
\theta_1  \\
\theta_2   \\
\theta_3  \\
\theta_4  \\
\phi_1   \\
\phi_3   \\
\end{pmatrix}
=
\begin{pmatrix}
2\pi \\
2\delta_{12} -2\pi \\
-2\delta_{12} +2\pi   \\
2\delta_{34} +2\delta_{12} -2\pi \\
\end{pmatrix}
\end{equation}
Consistency of this system requires 
\begin{equation}
2\delta_{34} +2\delta_{12} -2\pi =0
\end{equation}
or
\begin{equation}
\delta_{34} +\delta_{12} =\pi
\end{equation}
In terms of $\psi_{12}$ and $\psi_{34}$  
\begin{equation}
\pi-\psi_{12}+\pi-\psi_{34}=\pi
\end{equation}
Hence $\psi_{12}+\psi_{34}=\pi$, which is the familiar result that the opposite angles of a cyclic quadrilateral are supplementary.

\subsection{General Theorem}
The following general theorem may be proved by an analogous approach to that employed in the above section.

\begin{theorem}
\label{generaltheorem}
Given a cyclic 2n-gon with winding number $W$ about the circumcircle's center, and $n$ ordered pairs $(a_1,b_1) \ldots (a_n,b_n)$ such that 

$\{ a_i \} \cup \{b_i\} = \{ 1 \ldots 2n \}$,
$b_i>a_i$
and $b_i-a_i$ is odd for each $i$

\begin{equation}
\label{generaltheoremeqn}
\sum_{i=1}^n (-1)^{b_i}\delta_{a_ib_i}=W\pi
\end{equation}
\end{theorem}

\begin{proof}
Let $P_{i,j}=-2$ where $j=a_i$ or $j=b_i$ and 0 otherwise, let $Q_1=2\pi W$. For $i>1$, let $Q_j=2\delta_{a_i,b_i}$ if $j=b_i$ and $0$ otherwise.

Analogous to equation (\ref{eq:matrixeq}) we have:
\begin{equation} \label{eq:matrixeq2}
\begin{pmatrix}
1 & 0 & \cdots & 0 & 1 &P_{1,1} & \cdots & P_{n,1}   \\
1 & 1 & \cdots &0 &  0 &P_{1,2} & \cdots & P_{n,2}   \\
\vdots & \vdots & \ddots & \vdots & \vdots & \vdots &\vdots   &\vdots \\
0 & 0 & \cdots & 1 & 1 & P_{1,2n} & \cdots &P_{n,2n}     \\
\end{pmatrix}
\begin{pmatrix}
\theta_1  \\
\theta_2   \\
\vdots \\
\theta_{2n-1}  \\
\theta_{2n}  \\
\phi_1   \\
\vdots  \\
\phi_n   \\
\end{pmatrix}
=
\begin{pmatrix}
Q_1 \\
Q_2 \\
\vdots  \\
Q_{2n} \\
\end{pmatrix}
\end{equation}
Triangulation yields a matrix equation analogous to that of (\ref{eq:triangularmatrix}(
\begin{equation} \label{eq:matrixeq3}
\begin{pmatrix}
\vdots & \vdots & \vdots & \vdots & \vdots &\vdots    \\
0& \cdots & 0 &  R_1  & \cdots &   R_n \\
\end{pmatrix}
\begin{pmatrix}
\theta_1  \\
\vdots \\
\theta_{2n}  \\
\phi_1   \\
\vdots  \\
\phi_n   \\
\end{pmatrix}
=
\begin{pmatrix}
\vdots  \\
S \\
\end{pmatrix}
\end{equation}
where 
\begin{equation}
R_i= \sum_{j=1}^{2n}(-1)^jP_{i,j}
\end{equation}
and 
\begin{equation}
S= \sum_{j=1}^{2n}(-1)^jQ_j
\end{equation}
Hence
\begin{equation}
R_i= (-1)^{b_i}+(-1)^{a_i}
\end{equation}
and
\begin{equation}
S=\sum_{i=1}^n (-1)^{b_i}2\delta_{a_ib_i}-2W\pi
\end{equation}
As $b_i-a_i$ is odd,  $R_i=0$ for all $i$. 
Hence for consistency $S=0$.
Hence the result.
\end{proof}
\section{Automated Discovery of Cyclic Polygon Theorems}
We describe a mechanized process for automatic theorem generation for angles in cyclic polygons.  Each theorem will establish a relationship between angles in a cyclic polygon.  The process breaks down into three steps.  First a set of side pairs is chosen which satisfy the criteria of Theorem \ref{generaltheorem}.  Secondly, a specific location is decided for the vertices of the polygon.and a geometry diagram created.  Angles between the side pairs are given names and drawn on the diagram.  Thirdly, (\ref{generaltheoremeqn}) is expressed in terms of these angle names, yielding the conclusion of our theorem statement. 

\subsection{Choosing Side Pairs}
\begin{figure}[h]%
\centering
\includegraphics[width=0.95\textwidth]{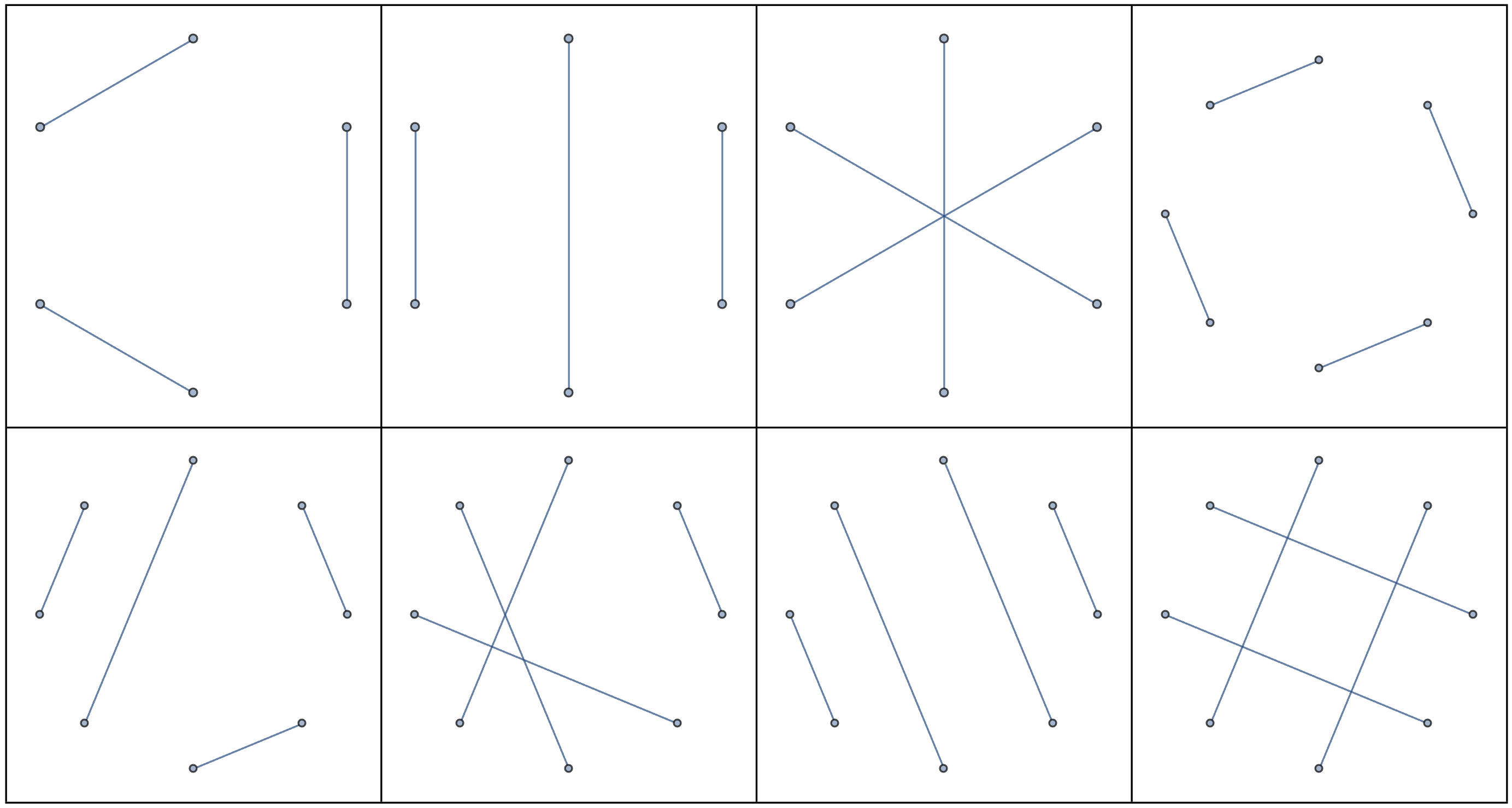}
\caption{Graphs representing possible conforming sets of edge pairs for $n=3, 4$.  Vertices on these graphs correspond to polygon sides.  Edges indicate the side pairs whose angles will be specified}
\label{pairdiagrams}
\end{figure}

In order to use theorem \ref{generaltheorem} to create cyclic polygon theorems, we need to first find a set of ordered pairs $(a_1,b_1) \ldots (a_n,b_n)$ such that 
$\{ a_i \} \cup \{b_i\} = \{ 1, \ldots ,2n \}$ and $|b_i-a_i|$ is odd for each $i$

We then sort the elements so that $a_i<b_i$.
Each pair $(a_i,b_i)$ contains an odd and an even integer.  Hence any set of pairs corresponds to a 1-1 mapping between  $A=\{1,3,\ldots ,2n-1\}$ and $B=\{2,4,\ldots ,2n\}$.  A random such mapping is selected.  Applying theorem \ref{generaltheorem} yields an expression for a linear combination  of the differences in direction $\delta_{a_ib_i}$.  

The number of such mappings is $n!$.  However, when considered as pairings between sides of a cyclic polygons, there are many which are simply rotated or reflected images.  The number of non-isomorphic patterns for $n=2$  to $7$ are $1, 3, 5, 17, 53, 260$ \cite{SequenceForSidePairs}.  Figure \ref{pairdiagrams} depicts these patterns for $n$  equal to 3 and 4.

\subsection{Diagram Creation}
A diagram should be constructed containing the circle and the cyclic polygon.  Points $q_{a_ib_i}$ should be computed  and angles  $\psi_{a_ib_i}$ should be marked on the diagram and named.  We need to relate the named angles: $\psi_{a_ib_i}$ to the angles treated in our theorem: $\delta_{a_ib_i}$.

One straightforward approach is to create a polygon whose vertices are arranged clockwise around the circle, making a single circuit, but at non-regular intervals.  Constraints may be set on the  location of the polygon vertices to ensure that the angle between chosen side-pairs is not too small for their intersection to appear on the diagram.  This intersection would be displayed, along with a named angle (Figure \ref{decagonexample}).
\subsection{Geometric Angle Conclusion}
The relation between $\delta_{ij}$ and $\psi_{ij}$ depends on whether the circle center is on the same side of $w_i$ and $w_j$ and whether a rotation from $w_i$ to $w_j$ is clockwise or counter-clockwise.

When the polygon is convex and has a winding number of 1 around the circumcenter, the formula is as follows

Let $s=sgn(w_i\wedge w_j)$ and
 
\begin{equation}
\delta_{ij}=\pi-s\cdot (\psi_{ij})
\end{equation}

\section{Example}
\begin{figure}[h]%
\centering
\includegraphics[width=0.28\textwidth]{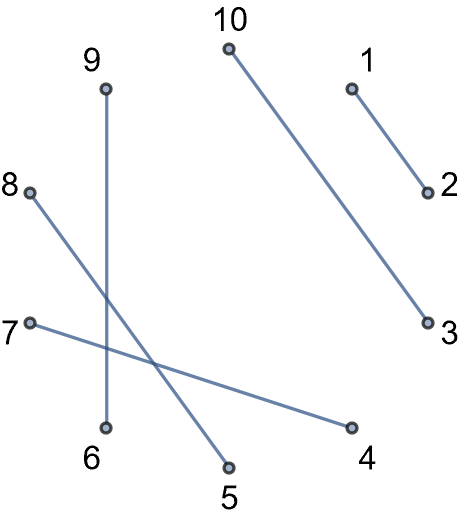}
\caption{Graph for the set of side pairings $\{(1,2),(3,10),(4,7),(5,8),(6,9)\}$}
\label{decagonpairings}
\end{figure}
As an example, we start with one of the sets of pairings for the decagon: $\{(1,2),(3,10),(4,7),(5,8),(6,9)\}$ (Figure \ref{decagonpairings}).  

Applying Theorem \ref{generaltheorem} to this set of side pairings, and assuming a winding number of 1 gives the following:
\begin{equation}
\label{decagoneqn}
\delta_{1,2}+\delta_{3,10}-\delta_{4,7}+\delta_{5,8}-\delta_{6,9}=\pi
\end{equation}

A cyclic decagon is drawn, where the vertices are placed at random subject to constraints that the polygon does not self-intersect, the sides are not too small, and the winding number is 1.  Angles between the five side pairs are marked and named.  In the case of non-contiguous sides, the sides are extended in the diagram to their intersection (Figure \ref{decagonexample}).
\begin{figure}[h]%
\centering
\includegraphics[width=0.55\textwidth]{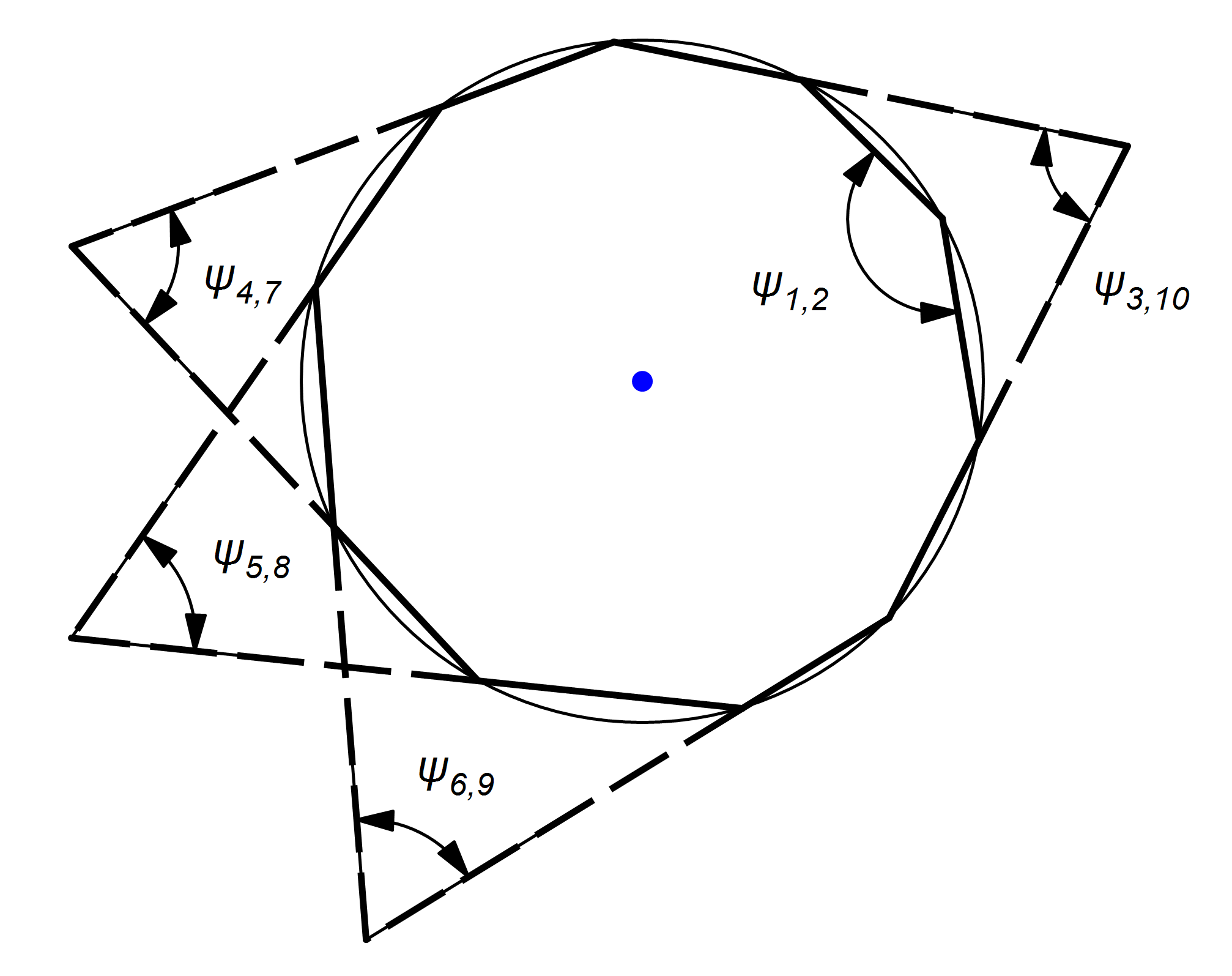}
\caption{Cyclic decagon with angles marked for side pairings $\{(1,2),(3,10),(4,7),(5,8),(6,9)\}$}
\label{decagonexample}
\end{figure}
As their angles are counter-clockwise, we have:
\begin{equation}
\begin{matrix}
\delta_{1,2}=\pi-\psi_{1,2}\\
\delta_{4,7}=\pi-\psi_{4,7}\\
\delta_{5,8}=\pi-\psi_{5,8}\\
\delta_{6,9}=\pi-\psi_{6,9}\\
\end{matrix}
\end{equation}
However, the angle for side pair (3,10) is clockwise, so
\begin{equation}
\label{cwangle}
\delta_{3,10}=\pi+\psi_{3,10}\\
\end{equation}
Substituting these values into (\ref{decagoneqn}) gives the following
\begin{equation}
\label{decagoneqnpsi}
\pi-\psi_{1,2}+\pi+\psi_{3,10}-(\pi-\psi_{4,7})+\pi-\psi_{5,8}-(\pi-\psi_{6,9})=\pi
\end{equation}
which can be presented as this theorem conclusion
\begin{equation}
\label{decagonconclusion}
\psi_{3,10}+\psi_{4,7}+\psi_{6,9}=\psi_{1,2}+\psi_{5,8}
\end{equation}
We note that it would be possible to draw a polygon such that the angle between lines 3 and 10 is clockwise.  In this case (\ref{cwangle}) becomes
\begin{equation}
\label{cwangle}
\delta_{3,10}=\pi-\psi_{3,10}\\
\end{equation}
and the theorem conclusion (\ref{decagonconclusion}) becomes
\begin{equation}
\label{decagonconclusion2}
\psi_{4,7}+\psi_{6,9}=\psi_{1,2}+\psi_{5,8}+\psi_{3,10}
\end{equation}

\section{An Automated Problem Generator}
\begin{figure}[h]%
\centering
\includegraphics[width=0.99\textwidth]{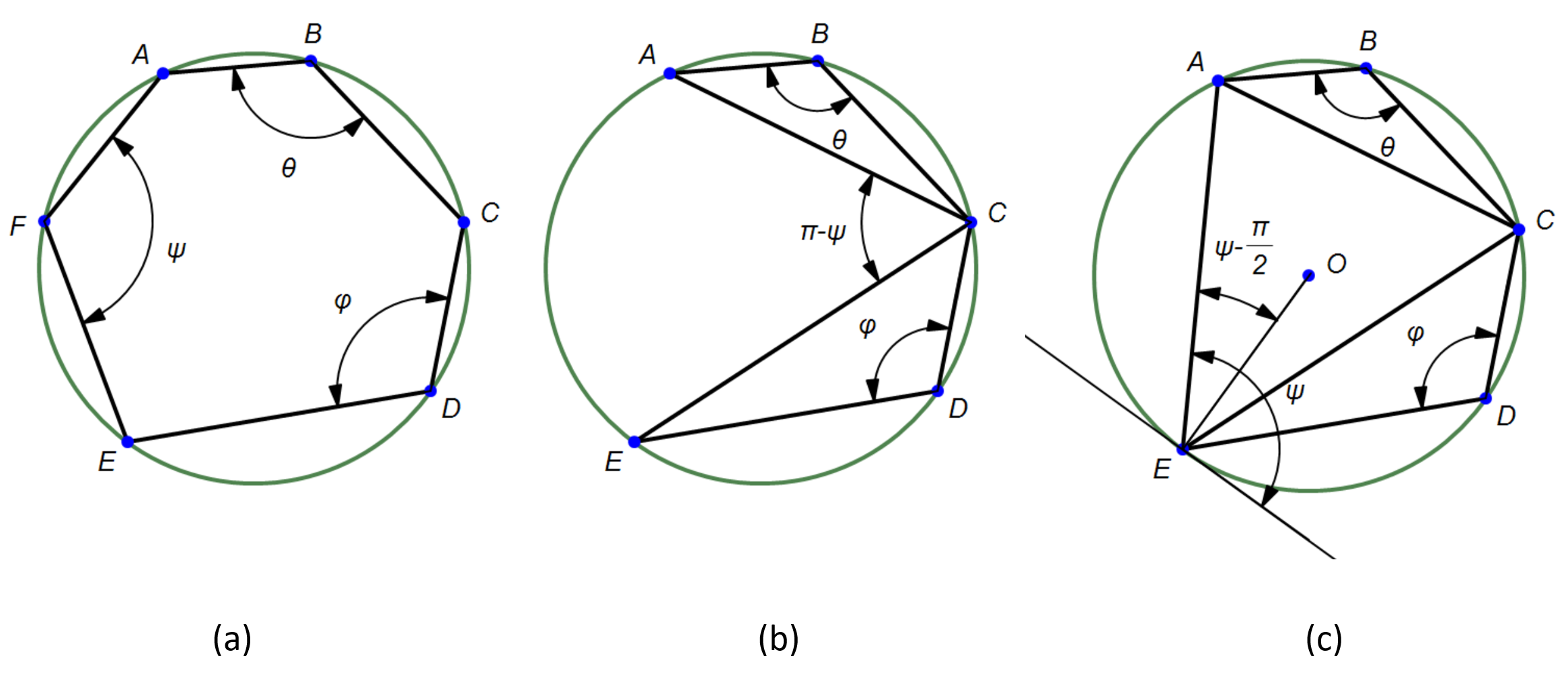}
\caption{(a) Given convex cyclic hexagon $ABCDEF$, $\angle ABC + \angle CDE + \angle EFA = 2\pi$. (b) Given convex cyclic pentagon $ABCDE$,  $\angle ABC + \angle CDE - \angle ECA = \pi$. (c) Given convex cyclic hexagon $ABCDE$ with center $O$,  $\angle ABC + \angle CDE + \angle OEA = \frac{3\pi}{2}$. }
\label{pointmerge}
\end{figure}

An automated problem generator based on the above approach was created  \cite{problemGenerator}, and introduces further elaborations.  First, rather than the vertices of the polygon lying in order around the circle, we allow any permutation of the vertices.  With any but the identity permutation, the polygon will be self-intersecting and will, in fact, be made up of some subset of the sides and diagonals of a convex polygon.  

A theorem for a polygon with an odd number of sides may be constructed from a polygon with one more side by merging a pair of points in the larger polygon. For example, figure \ref{pointmerge} (a) shows a theorem linking three alternating angles of cyclic polygon $ABCDEF$.  
Figure \ref{pointmerge}(b) shows the result of merging points $F$ and $C$ of the hexagon.  Edges $EF$ and $FA$ now become the diagonals $EC$ and $CA$ of a cyclic pentagon $ABCDE$, and the theorem relates two angles of the pentagon and an angle between two of its diagonals.
If the points which are merged share an edge of the larger polygon, that edge will be replaced by a line joining the center of the circle to the merged point, and the theorem statement will be altered by $\frac{\pi}{2}$.
Figure \ref{pointmerge}(c) shows the result of merging points $F$ and $E$ of the hexagon $ABCDEF$.  Angle $AFE$ becomes, in the limit, the angle between $AE$ and the tangent to the circle at $E$, which differs from $OE$ by $\frac{\pi}{2}$ (where $O$ is the circle center).

In addition to presenting the theorems, our problem generator \cite{problemGenerator} can create a human readable proof by propagating a handful of simple angle generation rules.  
A single step of the proof generator proceeds by finding all the angles determined by the application of the following set of rules to the known angles:
\begin{enumerate}
	\item If $AD$ is between $AB$ and $AC$ then $\angle BAC = \angle DAB + \angle DAC$.
    \item The angles of a triangle add to $\pi$.
    \item Two angles which form a line add to $\pi$.
    \item Two angles on the same side of the same chord of a circle are the same.
     \item Opposite angles of a cyclic quadrilateral add to $\pi$.
     \item The angle at the center of a circle is twice the angle on the same chord at the circumference.
     \item If $AB$ is a chord of circle centered $O$ then $\angle OAB = \angle OBA = \frac{1}{2}(\pi-\angle AOB)$.
\end{enumerate}
\begin{figure}[h]%
\centering
\includegraphics[width=0.65\textwidth]{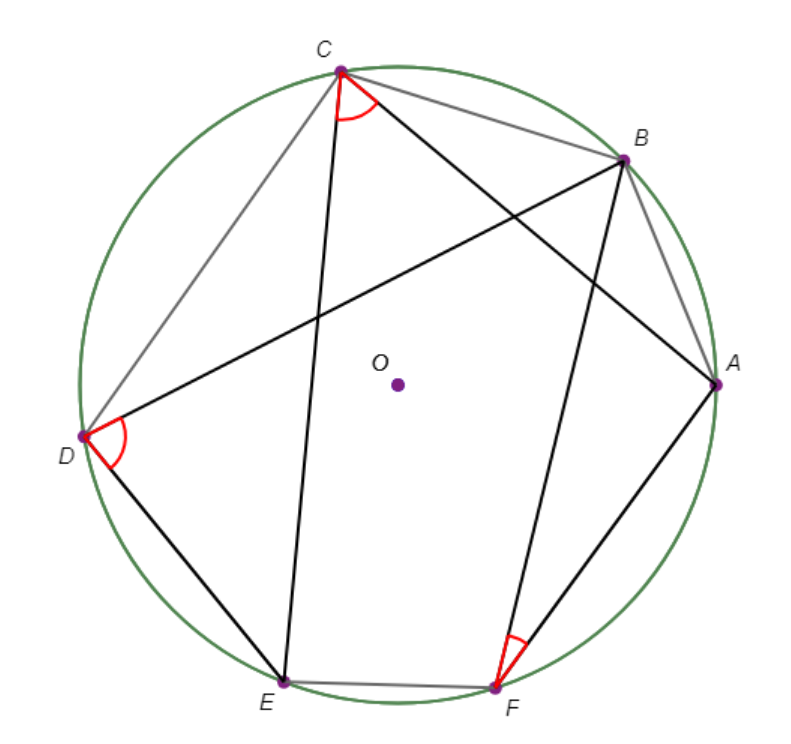}
\caption{Automatically generated diagram with problem statement: Let $ABCDEF$ be a cyclic hexagon with center $O$. Prove that $BDE = ACE+AFB$}
\label{problemgenerated}
\end{figure}
Our theorems can all be expressed as a linear combination of the angles and $\pi$.  The result of each angle generation rule is a linear combination of its inputs and $\pi$.  Hence a convenient way of storing the computed value of any angle is as a vector of coefficients of the given angles and $\pi$. 
The proof generator is started with the known angles and proceeds to compute values for all angles which can be computed from those by application of the above rules (breadth first).  
The angle value is stored, as a vector of coefficients of the known angles and $\pi$.  Along with the value, a reference is stored to the rule which was fired, and the arguments it was applied to.
If the process finds a new value for a given angle, this new value can be equated to the given angle yielding a linear expression. The sequence of rule applications used to reach the final expression may be turned into the words of a proof.

For example, in figure \ref{problemgenerated}, we set $\angle BDE=(1,0,0,0)$, $\angle ACE=(0,1,0,0)$ , $\angle AFB=(0,0,1,0)$.

Applying (4) to $\angle AFB$ gives $\angle ACB=(0,0,1,0)$.

Applying (4) to $\angle BDE$ gives $\angle BCE=(1,0,0,0)$.

Applying (1) to $\angle BCE$ and $\angle ACE$ gives $\angle ACB=(1,-1,0,0)$.
 
But we already have $\angle ACB=(0,0,1,0)$ so our linear relation can be expressed as $(1,-1,-1,0)=0$.

This would result in the following proof.

Let $BDE=x$. Let $ACE=y$. $Let AFB=z$. 

As $AFB$ and $ACB$ stand on the same chord, $ACB=AFB$, so $ACB=z$.

As $BDE$ and $BCE$ stand on the same chord, $BCE=BDE$, so $BCE=x$.

As $ACE=y$, $ACB=x-y$.

But $ACB=z$, so $x-y=z$, or $x=y+z$, or $BDE=ACE+AFB$.

If the proof generator does not run to completion, additional geometry is added to the diagram in two phases.  First, missing line segments are added, which allow additional angles to be derived (angles are only computed when the line segments joining their vertices are present). A second step is for an additional angle to be specified as known, and given a name.  As long as the new angle is not in the span of the known angles, it cannot show up in the eventual linear relation, and hence is guaranteed to simplify out in the course of the proof.
\begin{figure}[h]%
\centering
\includegraphics[width=0.65\textwidth]{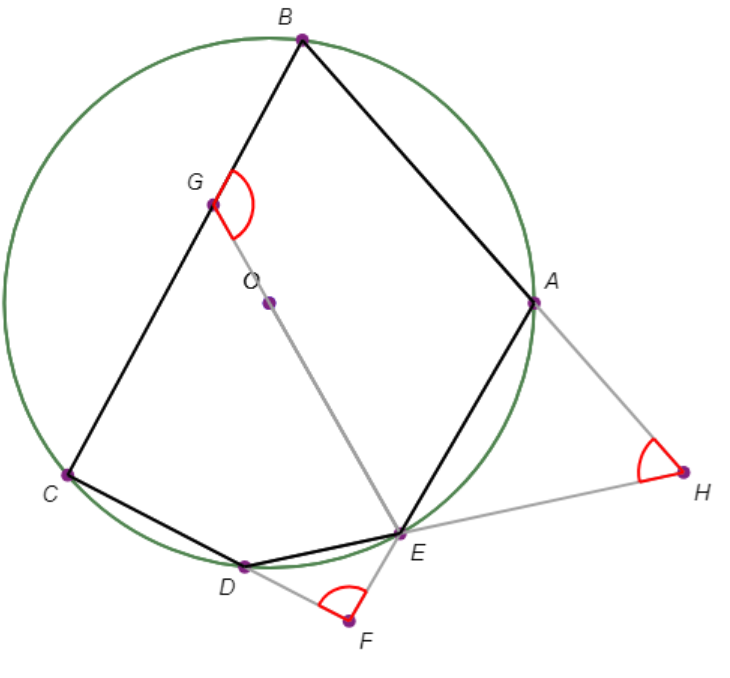}
\caption{Let $ABCDE$ be a cyclic pentagon with center $O$. Let $F$ be the intersection of $AE$ and $DC$. Let $G$ be the intersection of $OE$ and $CB$. Let $H$ be the intersection of $ED$ and $BA$. Prove that $BGE+AHE = DFE+90$.}
\label{problemgenerated2}
\end{figure}

 For example the following is the proof generated for the theorem stated in figure \ref{problemgenerated2}. 
 
 \begin{theorem}
 Let $ABCDE$ be a cyclic pentagon with center $O$. Let $F$ be the intersection of $AE$ and $DC$. Let $G$ be the intersection of $OE$ and $CB$. Let $H$ be the intersection of $ED$ and $BA$. $BGE+AHE = DFE+90$.
\end{theorem}
\begin{proof}
Draw line $BE$.

Let $\angle DFE=x$. Let $\angle BGE=y$. Let $\angle AHE=z$. 

Let $\angle AEH=w$.

As$\angle  AHE=z$, $\angle EAH=180-z-w$.

As $\angle EAH=180-z-w$, $\angle EAB=z+w$.

As $\angle AEH=w$, $\angle HEF=180-w$.

As $\angle FEH=180-w$, $\angle FED=w$.

As $\angle DEF=w$, $\angle EDF=180-x-w$.

As $\angle EDF=180-x-w$, $\angle EDC=x+w$.

As $CDEB$ is a cyclic quadrilateral, $\angle CBE=180-\angle CDE$, so $\angle CBE=180-x-w$.

As $\angle EBG=180-x-w$, $\angle BEG=x+w-y$.

As triangle $BEO$ is isosceles, $\angle BOE=2y-2x-2w+180$.

As $\angle BOE$ is at the center of a circle on the same chord, but in the opposite direction to $\angle BAE$, $\angle BOE=360-2BAE$, so $\angle BAE=x+w-y+90$.

But $\angle BAE=z+w$, so $x+w-y+90=z+w$, or $x+90=y+z$, or $\angle DFE+90=\angle BGE+\angle AHE$.
\end{proof}

Automated proofs can be simply adapted to provide step-by-step solutions to problems framed in terms of determining an unknown angle, either numerically or algebraically.

Our generator can be set to create problem collections.  It is important for such collections to eliminate problems which are simply rotated or reflected duplicates.  For a hexagon, for example, we see from figure \ref{pairdiagrams} that there are 3 different pair patterns.  When we combine this with the different permutations of 6 points, and remove rotated and reflected duplicates, we end up with 49 distinct diagrams.

For the pentagon, we can choose any pair of points to merge, but again need to avoid duplication by keeping track not only of the diagrams which have been produced, but also of their rotated and reflected images. The pentagon yields 54 distinct diagrams. \cite{completeSet} shows the complete collection of pentagon and hexagon theorems generated in this way.

A heptagon theorem may be generated from an octagon theorem by merging a pair of points.  A hexagon theorem may be formed from the heptagon theorem by merging a further pair of points.  Admitting such ``four angle'' hexagon theorems expands the number of possibilities to hundreds.  Admitting octagons and heptagons expands the possible theorem count to thousands.   \cite{200Problems}  is a collection of 200 random theorems and their proofs generated by  \cite{problemGenerator} .

\section{Conclusion}

The angles considered in this paper are not the full angles which have their own place in automated deduction in geometry \cite{chou1996automated2}, and which are largely impervious to changes in the diagram caused by simply positioning the points in different locations. The angles we use are well defined: an angle $ABC$ means the non-reflex undirected angle defined by the line segments $AB$ and $BC$. While the angle is well defined, its relationship to the directed angle, which is the basis of, for example, (\ref{decagoneqn}) is highly dependent on the relative location of the points on the diagram.  When an angle passes through $\pi$, the relation between the ordered angle and the unordered non-reflex angle changes and, in our setting, the theorem statement. This is both a blessing and a curse.  A curse because it is difficult to give a succinct general theorem statement without a complicated set of side conditions.  In the context of automated problem generation, however, the curse becomes a blessing, as it proliferates the number of distinct examples which may be generated.  

From a pedagogical standpoint, the problems generated have several advantages.  They can all be solved by applying a small set of geometric facts, however some ingenuity in applying them is called for.  Algebraically, the problems require only the manipulation of linear equations.  Problems can be phrased as the determination of unknown quantities, either numerically or symbolically, or as proof problems.  Finally the availability of machine generated human-readable proofs gives a convenient means of scaffolding.

\bibliographystyle{eptcs}
\bibliography{mybibliography}

\end{document}